\documentclass{llncs}
\usepackage{makeidx}  
\begin{document}

\title{Computation of the Adjoint Matrix}
\titlerunning{Computation of Adjoint Matrix}  
\author{Alkiviadis Akritas\inst{1} \and Gennadi Malaschonok\inst{2} 
\thanks{Paper was published in: Computational Science - ICCS 2006: 6th Int. Conf., UK, May 28-31, 2006. Proc., Part II, LNCS 3992, Springer Berlin, 2006, 486-489}}
\authorrunning{A. Akritas and G.Malaschonok}   
\tocauthor{Alkiviadis Akritas (University of Thessaly), 
Gennadi Malaschonok (Tambov State University)}
\institute{University of Thessaly, Department of
 Computer and Communication Engineering, GR-38221 Volos, Greece\\
\email{akritas@uth.gr}
\and
Tambov State University, Laboratory for Algebraic Computations, \\ 
Internatsionalnaya 33, 392622 Tambov, Russia, \\
\email{malaschonok@math-iu.tstu.ru}
}

\maketitle               

\begin{abstract}
The best method for computing the adjoint matrix of an order $n$ matrix in an arbitrary commutative ring requires 
\linebreak
$O(n^{\beta+1/3}\log n \log \log n)$ operations, provided the complexity of the algorithm for multiplying two 
matrices is 
$\gamma n^\beta+o(n^\beta)$.  
For a commutative domain -- and under the same assumptions -- the complexity of the best method is 
${6\gamma n^\beta}/{(2^{\beta}-2)}+o(n^\beta)$. 
In the present work a new method is presented for the computation of the adjoint matrix in a commutative domain.  
Despite the fact that the number of operations required is now 1.5 times more, than that of the best method, 
this new method  permits a better parallelization of the computational process and may be successfully employed 
for computations in parallel computational systems.
\end{abstract}

\newcommand{\Dmtr}[4] 
{\left|\begin{array}{rr}#1 & #2\\#3 & #4\end{array}\right|}

\section{Statement of the problem}

\hspace*{0.5 cm} The adjoint matrix is a transposed matrix of algebraic complements.  
If the determinant of the matrix is invertible, then the inverse 
matrix may be computed as the adjoint matrix divided by the determinant.  
The adjoint matrix of a given matrix  $A$ will be denoted by  $A^*$: 
$A^*=\det(A) A^{-1}$.

The best method for computing the adjoint matrix of an order $n$ matrix in an arbitrary commutative ring requires 
$O(n^{\beta+1/3}\log n \log \log n)$ operations (see \cite{Kalt92} and \cite{KaVi01}).  
For a commutative domain the complexity of the best method is 
${6\gamma n^\beta}/{(2^{\beta}-2)}+o(n^\beta)$ (see \cite{Ma00}).
It is asssumed that the complexity of the algorithm for multiplying two matrices is 
$\gamma n^\beta+o(n^\beta)$.  
 
In a commutative domain the algorithm is based on applications of determinant identities \cite{Ma00}, \cite{Ma83}.
It generalizes in a commutative domain the following formula for the inverse matrix ${\cal A}^{-1}$:
$$ 
{\cal A}^{-1}=
\left(\begin{array}{cc} I&-A^{-1}C\\0&I \end{array}\right)
      \left(\begin{array}{cc} I&0\\0& (D-BA^{-1}C)^{-1} \end{array}\right)
 \left(\begin{array}{cc} I&0\\-B&I                 \end{array}\right)
 \left(\begin{array}{cc} A^{-1}&0\\0&I                  \end{array}\right),
$$
where  ${\cal A}= \left(\begin{array}{cc} A&C\\B&D \end{array}\right)$ --
is an invertible matrix with invertible block $A$.

In the present work a new method is proposed for the computation of the adjoint matrix in a commutative domain.  
Despite the fact that the number of operations required is now 1.5 times more, than that of the algorithm 
described in [1], this new method  permits a better parallelization of the computational process. 

This new method generalizes in a commutative domain the following factorization of the inverse matrix ${\cal A}^{-1}$:
$$ 
{\cal A}^{-1}=
\left(\begin{array}{cc} I&-A^{-1}C\\0&I \end{array}\right)
      \left(\begin{array}{cc} I&0\\0& (B^{-1}D-A^{-1}C)^{-1} \end{array}\right) 
 \left(\begin{array}{cc} I&0\\-I&I                 \end{array}\right)
 \left(\begin{array}{cc} A^{-1}&0\\0&B^{-1}        \end{array}\right).
$$

The second section is devoted to the proof of the determinant identity of column replacement, which is used as 
the basis of the proposed method for computing the adjoint matrix.  In the third section additional theorems are 
proved, which are fundamental for the new method.  In the fourth section the algorithm and a small example are 
presented for the computation of the adjoint matrix.  Finally, in the fifth section, a discussion is presented 
of the algorithm along with its advantages.

\section {Identity of column replacement}
Let $B$ be a matrix of order $n$ and assume two different columns fixed.
We denote by $B_{\{x,y\}}$ the matrix which is obtained from $B$ after replacing the two fixed columns by the 
columns $x$ and $y$, respectively.

\begin{theorem} {(Identity of column replacement.)}

For every matrix $B\in R^{n \times n}$ and columns $a,b,c,d\in R^n$
the following identity holds
$$
        \det B_{\{ab\}} \det B_{\{cd\}}=
\Dmtr {\det B_{\{ad\}}} {\det B_{\{db\}}}
      {\det B_{\{ac\}}} {\det B_{\{cb\}}}
\eqno(1)
$$
\end{theorem}
\begin{proof}

Let $O$ denote the zero matrix of order $n$, and let
$o\in R^n$ denote the zero column.  Looking at the determinant equation
$$ \Dmtr {  B_{\{ab\}}} {  B_{\{oo\}}}
{  O_{\{ao\}}} {  B_{\{cd\}}}
=
\Dmtr {  B_{\{ob\}}} {  O_{\{-c,-d\}}}
{  O_{\{ao\}}} {  B_{\{c,d\}}}.
$$
we observe the following: The determinant on the right is obtained from the determinant on the left by 
subtracting the second block row from the first.  
Moreover, using Laplace's expansion theorem, every determinant may be expanded by the first 
$n$ rows according to the formula
$$
{\det B_{\{ab\}}}{\det B_{\{cd\}}}=
(-1)^{n}\det(B_{\{*,b\}},-c) \det (a,B_{\{*,d\}})+
$$
$$
+(-1)^{n+i+j}\det(B_{\{*,b\}},-d) \det (a,B_{\{c,*\}})
  \eqno(2)
$$
where $i$ and $j$ are the numbers of the fixed columns $c$ and $b$ in the matrix
$B_{\{c,b\}}$, and the matrices $B_{\{*,b\}}$ and $B_{\{c,*\}}$ are obtained from
$B_{\{c,b\}}$ by deleting rows $c$ and $b$, respectively.

On the right-hand side of (2) we apply the identities

\noindent
$\det(B_{\{*,b\}},-c)= (-1)^{n-i}\det B_{\{-c,b\}}$,

\noindent
$\det(a,B_{\{*,d\}}) = (-1)^{i-1}\det B_{\{a,d\}} $,

\noindent
$\det(B_{\{*,b\}},-d)= (-1)^{n-i}\det B_{\{-d,b\}}$,

\noindent
$\det(a,B_{\{c,*\}}) = (-1)^{j-1}\det B_{\{c,a\}}$ $=
(-1)^{j}\det B_{\{a,c\}}$,

\noindent
and obtain identity (1).
\qed
\end{proof}
 
For example, for the matrix of order 2 the identity of column replacement is as follows
$$
\Dmtr a c  b d \Dmtr x u  y v=
\Dmtr a u  b v \Dmtr x c  y d
-\Dmtr a x  b y \Dmtr  u  c v d.
$$


\section {Fundamental theorems}

\begin{theorem}
Let $R$ be a commutative domain, 
${\cal A}=
\left(\begin{array}{cc} A&C\\B&D \end{array}\right)$ a matrix of order $2n$ over $R$,  $A, B, C, D$ square blocks,
$\alpha=\det A \neq 0$,  $\beta=\det B \neq 0$, $F=\alpha B^*D-\beta A^*C$.

Then every minor of order $k$ of the matrix $F$  ($k\leq n$) 
is divisible by $(\alpha\beta)^{k-1}$
and the following identity holds
$$
\det F = (\alpha\beta)^{n-1} \det {\cal A}.
\eqno (3)
$$ 
\end{theorem}
\begin{proof}

Let ${\delta_{i,j}}$ be the determinant of matrix $A$ after replacing its column
 $i$ by column  $j$ of matrix $C$, and let
${\delta'_{i,j}}$ be the determinant of matrix $B$ after replacing its column 
 $i$ by column  $j$ of matrix $D$. 

It is then obvious that  $A^*C=(\delta_{i,j})$ and 
$B^*D=(\delta'_{i,j})$, $i,j=1,\ldots,n$, and that the elements of matrix $F=(f_{i,j})$ can be expressed as 
$f_{i,j}=\alpha \delta'_{i,j} - \beta \delta_{i,j}$.

Let us first examine an arbitrary minor of order two of the matrix $F$: 
$$
\Dmtr {f_{i,j}} {f_{i,q}} {f_{p,j}} {f_{p,q}} =
\alpha^2\Dmtr {\delta'_{i,j}} {\delta'_{i,q}} {\delta'_{p,j}} {\delta'_{p,q}} 
+\beta^2\Dmtr {\delta_{i,j}} {\delta_{i,q}} {\delta_{p,j}} {\delta_{p,q}} 
$$
$$
-\alpha \beta ( {\delta'_{i,j}} {\delta_{p,q}} + {\delta_{i,j}} {\delta'_{p,q}} 
-{\delta'_{i,q}} {\delta_{p,j}} -{\delta_{i,q}} {\delta'_{p,j}} 
)
$$
From the identity of column replacement it follows that
$$
\Dmtr {\delta'_{i,j}} {\delta'_{i,q}} {\delta'_{p,j}} {\delta'_{p,q}}=
\beta {\delta'}_{pq}^{ij},
\ \Dmtr {\delta_{i,j}} {\delta_{i,q}} {\delta_{p,j}} {\delta_{p,q}}=\alpha {\delta}_{pq}^{ij}. 
$$

Here ${\delta}_{pq}^{ij} $ (and also ${\delta'}_{pq}^{ij}$) is the determinant of the matrix $A$ ($B$) after 
replacing columns $i$  and $p$ by columns $j$ and $q$ of the matrix $C$ ($D$).

Consequently, every minor of order two of the matrix $F$ is divisible by $\alpha \beta$.

We next examine an arbitrary minor of order  $k$ of the matrix $F$ ($ 2 < k \leq n$).
Moreover, let this be the left upper corner minor of matrix  $F$. 
We denote by $F'$ the matrix of order $k$ corresponding to this minor,
 and by $G$ the matrix of order $k-1$, which is formed by those minors of order two of the matrix $F'$, in which there 
appears the corner element  $f_{11}$.
 
Then we can write Sylvester's determinant identity 
$$
 \det(F') (f_{11})^{k-2}=\det(G).
$$

Every element of the matrix $G$ is divisible by $\alpha \beta$, since it a minor of order two of the matrix $F$. 
Therefore,  $\det(G)$ is divisible by $(\alpha \beta)^{k-1}$. 
The element 
$f_{11}=\alpha \delta'_{11} - \beta \delta_{11}$, considered as a polynomial of elements of the matrix 
${\cal A}=(a_{i,j})$, does not have common multiples with $\alpha \beta$.
Consequently, $ \det(F')$ is divisible by $(\alpha \beta)^{k-1}$.

To prove the last claim of the theorem we examine the matrix identity
$$ 
\left(\begin{array}{cc} I&0\\-\beta  I&\alpha I \end{array}\right)
\left(\begin{array}{cc} A^*&0\\ 0 & B^*         \end{array}\right)
\left(\begin{array}{cc} A&C\\B&D                \end{array}\right)=
\left(\begin{array}{cc} \alpha I & A^*C \\ 0&F  \end{array}\right).
$$
The right-hand side is the product of the matrices on the left-hand side.  Corresponding to this matrix 
identity we have the determinant identity
$ \alpha^n \cdot \alpha^{n-1} \beta^{n-1} \cdot \det{\cal A}=\alpha^n \det F$, from which follows  (3).
\qed
\end{proof}

\begin{theorem}
Let $R$ be a commutative domain, 
${\cal A}=
\left(\begin{array}{cc} A&C\\B&D \end{array}\right)$ a matrix of order $2n$ over $R$,  $A, B, C, D$ square blocks,
$\alpha=\det A \neq 0$,  $\beta=\det B \neq 0$ and $F=\alpha B^*D-\beta A^*C$.
 
Then $(\alpha \beta)^{-n+2}F^*\in R^n $ and 
 
$$ {\cal A}^*=
\left(\begin{array}{cc}  \alpha^{-1}|{\cal A}|I&- \alpha^{-2}\beta^{-1}A^*C\\0&\alpha^{-1}\beta^{-1}I\end{array}\right)
\left(\begin{array}{cc} I&0\\ 0 & (\alpha \beta)^{-n+2}F^*         \end{array}\right)
\left(\begin{array}{cc} I&0\\-\beta  I&\alpha I \end{array}\right)
\left(\begin{array}{cc} A^*&0\\ 0 & B^*         \end{array}\right).
 \eqno (4)
$$
\end{theorem}
\begin{proof}

Let us myltiply the matrix A from the left by the terms of the right-hand side of (4), sequentially, begining with
the last term. We obtain step by step the following:
$$
{\cal A}
\rightarrow
\left(\begin{array}{cc} \alpha I &A^*C \\ \beta I & B^*D  \end{array}\right)
\rightarrow
\left(\begin{array}{cc} \alpha I &A^*C \\ 0 & F  \end{array}\right)
\rightarrow
\left(\begin{array}{cc} \alpha I &A^*C \\ 0 & \alpha \beta|{\cal A}|I  \end{array}\right)
\rightarrow 
%
|{\cal A}|I.
$$
Note that the elements of the matrix $F^*$  are minors of order  $n-1$ of the matrix $F$ and according to Theorem 2 
they are divisible by $(\alpha\beta)^{n-2}$, and that we need to use identity (3).
\qed
\end{proof}

\begin{theorem}

Let $R$ be a commutative domain, $0\neq \gamma \in R$, 
${\cal A}=
\left(\begin{array}{cc} A&C\\B&D \end{array}\right)$ a matrix of order $2n$ $(n\geq 2)$ over $R$, such that every 
minor of order $k$ is divisible by 
$\gamma^{k-1}$, $A, B, C, D$ square blocks,
$\alpha=\gamma^{1-n}\det A \neq 0$,  $\beta=\gamma^{1-n}\det B \neq 0$, 
${\mathbf A}^*=\gamma^{2-n} A^*$, ${\mathbf B}^*=\gamma^{2-n} B^*$ and
$F=(\alpha \gamma^{-1}{\mathbf B}^*D-\beta \gamma^{-1}{\mathbf A}^*C)$.
 
Then,  
$$ \gamma^{2-2n}{\cal A}^*=
\left(\begin{array}{cc}  \alpha^{-1}\gamma^{1-2n}|{\cal A}|I&
      - (\alpha^{2}\beta\gamma)^{-1}{\mathbf A}^*C\\
0&(\alpha\beta\gamma)^{-1}I\end{array}\right)
\times
$$
$$
\left(\begin{array}{cc}       I&0     \\ 0 & (\alpha \beta)^{-n+2} F^*\end{array}\right) 
\left(\begin{array}{cc}       I&0     \\- \beta  I&  \alpha I         \end{array}\right)
\left(\begin{array}{cc}{\mathbf A}^*&0\\ 0 &  {\mathbf B}^*           \end{array}\right).
 \eqno (5)
$$
Here $\gamma^{2-2n}{\cal A}^*$ and the last three factors on the right-hand side of (5) are matrices over  $R$.
\end{theorem}
\begin{proof}

Let us myltiply the matrix A from the left by the terms of the right-hand side of (5) sequentially, begining with
the last term. We obtain step by step the following:

$$
{\cal A}
\rightarrow 
\left(\begin{array}{cc}\alpha\gamma I&{\mathbf A}^*C \\ \beta\gamma I& {\mathbf B}^*D\end{array}\right)
\rightarrow 
\left(\begin{array}{cc} \alpha\gamma I &{\mathbf A}^*C \\ 0 &\gamma F  \end{array}\right)
\rightarrow
\left(\begin{array}{cc} \alpha\gamma I &{\mathbf A}^*C \\ 0 & (\alpha \beta)^{-n+2}\gamma |F| I  \end{array}\right)
\rightarrow 
\gamma^{2-2n}|{\cal A}|I.
$$
Since $\alpha$ and $\beta$ are minors of order $n$, and the elements of the matrices 
$A^*$ and $B^*$ are minors of order $n-1$, then according to the conditions of this theorem and by Theorem 2 all 
divisions are exact.
\qed
\end{proof}

{\bf Consequence.} 
Let $R$ be a commutative domain, $0\neq \gamma \in R$, 
${\cal A}= 
\left(\begin{array}{cc} A&C\\B&D \end{array}\right)$ a matrix of order  $2n$ $(n\geq 2)$ over $R$, such that every
 minor of order $k$ is divisible by 
$\gamma^{k-1}$, $A, B, C, D$ square blocks,

$$
\alpha=\gamma^{1-n}\det A \neq 0, \  \beta=\gamma^{1-n}\det B \neq 0, \ \varphi=\gamma^{2-2n}\det {\cal A},
$$ 
$$
{\mathbf A}^*=\gamma^{2-n} A^*, \ {\mathbf B}^*=\gamma^{2-n} B^*,
$$
$$ 
F=(\alpha \gamma^{-1}{\mathbf B}^*D-\beta \gamma^{-1}{\mathbf A}^*C),
{\mathbf F}^*=(\alpha\beta)^{2-n}F^*,
$$ 
$$
H=\alpha^{-1}\gamma^{-1}{\mathbf F}^*{\mathbf A}^*, 
L=\beta^{-1}\gamma^{-1}{\mathbf F}^*{\mathbf B}^*, 
M=\alpha^{-1}{\mathbf A}^*C.
$$
  
Then, every minor of order  $s$ of matrix $F$ is divisible by $(\alpha\beta)^{s-1}$,
$$ 
\varphi=(\alpha\beta)^{1-n}\det F \ \ \hbox{and}
$$
$$ \gamma^{2-2n}{\cal A}^*
=
\left(\begin{array}{cc}
 \alpha^{-1}(\varphi{\mathbf A}^*+ MH)&  -\alpha^{-1}ML  \\
-H & L
\end{array}\right).
$$

\section{The algorithm}  

\hspace*{0.5 cm} Using the theorems we proved about the factorization of the adjoint matrix we now introduce the 
algorithm for computing it along with the determinant of a given matrix.

Let $R$ be a commutative domain, $0\neq \gamma \in R$, 
${\cal A}=
\left(\begin{array}{cc} A&C\\B&D \end{array}\right)$ a matrix of order $2n=2^N$ over $R$, such that every minor of order 
 $k$ is divisible by $\gamma^{k-1}$.  Moreover, we assume that all minors, on which a division is performed during the 
computation of the adjoint matrix, are non-zero. 

The inputs to the algorithm are the matrix $\cal A$ and the number $\gamma=1$.

The outputs from the algorithm are  
$\gamma^{1-2n}|{\cal A}|$ and  $ \gamma^{2-2n}{\cal A}^*$. 
Note here that the determinant of the matrix has been divided by  $\gamma^{2n-1}$, and that the adjoint matrix has 
been divided by  $\gamma^{2n-2}$.
\bigskip
 
\centerline {\bf{ Algorithm ParAdjD}}
\medskip

\centerline 
{\bf{\{ $\gamma^{1-2n}|{\cal A}|, \gamma^{2-2n}{\cal A}^*$
\}=ParAdjD($ {\cal A},  \gamma $)}}
\bigskip

\noindent
{\it Input:} ${\cal A}=\left(\begin{array}{cc} A&C\\B&D \end{array}\right)$, and $\gamma$.
$A,B,C,D\in R^{n\times n}$, $\gamma\in R$.

\noindent
{\it Output: } \{$\gamma^{1-2n}|{\cal A}|, \  \gamma^{2-2n}{\cal A}^*$ \}.  
\bigskip
\smallskip

\noindent
1. If the matrix  ${\cal A}$ is of order two, then 
 
\noindent
{\it output:} 
$$
\bigg\{
\gamma^{-1}(AD-BC), \
\left(\begin{array}{cc} D&-C\\-B&A \end{array}\right) \bigg\}.
$$

\noindent
otherwise, proceed to the next point.
 
\noindent
2.  
Concurrently compute

{\bf{\{ $\alpha, {\mathbf A}^* $\}=ParAdjD($A, \gamma$)}} and
{\bf{\{ $\beta, {\mathbf B}^* $\}=ParAdjD($B, \gamma$)}}.

\noindent
3. 
Concurrently compute

${N=\gamma^{-1}\mathbf B}^*D$ and $M=\gamma^{-1}{\mathbf A}^*C$, 
and then

$
F=\alpha N  -\beta M.
$

\noindent
4. Compute 

{\bf{\{ $  \varphi, {\mathbf F}^* $\}=ParAdjD($F, \alpha \beta$)}}.
 
\noindent
5. Concurrently compute

$\varphi'= \gamma^{-1}\varphi$, 
$H=\alpha^{-1}\gamma^{-1}{\mathbf F}^*{\mathbf A}^*$ and 
 $L=\beta^{-1}\gamma^{-1}{\mathbf F}^*{\mathbf B}^*$. 

\noindent
6. Concurrently compute

$H'=  \alpha^{-1}(\varphi' {\mathbf A}^*+ MH)$ and $L'=-\alpha^{-1}ML$.

\noindent
{\it Output:} 
$$
\bigg\{
\varphi', \
\left(\begin{array}{cc}   H'&  L'  \\
-H & L
 \end{array}\right) \bigg\}.
$$

\subsection{Example}
 
$$
{\cal A}=\left(\begin{array}{cc} A&C\\B&D \end{array}\right) 
 =\left(\begin{array}{cccc}
 0& 2&-2& 2\\
 1&-3& 1&-2\\
 3& 0&-3& 0\\
-1& 3&-1& 1  \end{array}\right), \ \gamma=1. 
$$

\noindent 
1. We concurrently compute 
$$
 \bigg\{ -2,\left(\begin{array}{cc} -3& -2\\ -1& 0 \end{array}\right) \bigg\}
 = \mbox{\bf ParAdjD}(A, 1)  \mbox{ and } 
$$
$$
  \bigg\{ 9, \left(\begin{array}{cc} 4&-2\\2&-2 \end{array}\right) \bigg\}
 = \mbox{\bf ParAdjD}(B, 1) .
$$
\noindent
2. We concurrently compute 

${N=\mathbf B}^*D = \left(\begin{array}{cc} -9& 0\\ -6& 3 
\end{array}\right)$
 and $M={\mathbf A}^*C=\left(\begin{array}{cc} -3& -2\\ -1& 0 
\end{array}\right)$

and then
$
F=\alpha N  -\beta M=\left(\begin{array}{cc} -18& 18\\ -6& 12 
\end{array}\right).
$
 
\noindent
3. We  compute 
$$
\bigg\{ 6,\left(\begin{array}{cc} 12& -18\\ 6& -18 \end{array}\right) \bigg\}
=\mbox{\bf ParAdjD}(F, (-2)\cdot 9) .
$$ 
4. We concurrently compute  $\varphi'= \gamma^{-1}\varphi=6$,

\noindent
$H=\alpha^{-1}\gamma^{-1}{\mathbf F}^*{\mathbf A}^* =
\left(\begin{array}{cc} 9& 12\\ 0& 6 \end{array}\right)$, 
 $L=\beta^{-1}\gamma^{-1}{\mathbf F}^*{\mathbf B}^* =
\left(\begin{array}{cc} 2& -6\\ 0& -6 \end{array}\right)$.

\noindent
5. We concurrently compute  

$H'={\alpha^{-1}(\varphi'{\mathbf A}^*+ MH)}=
\left(\begin{array}{cc} -9& -12\\ -6& -6 \end{array}\right)$

 and $L'= -\alpha^{-1} ML=
\left(\begin{array}{cc} 4 & -6 \\ 2 & 0 \end{array}\right)
$.

\noindent
{\it Output:} 
$$
\Bigg\{ 
{\varphi'}, \
\left(\begin{array}{cc}  H' & L' \\
-H & L
 \end{array}\right) \bigg\}=
\bigg\{6,
\left(
\begin{array}{cccc}  
-9& -12 & 4 & -6 \\
-6& -6 & 2 & 0 \\
-9& -12 & 2 & -6 \\
 0& -6 & 0 & -6  \end{array}\right) \Bigg\}.
$$
 
\section{Discussion on the algorithm}

We now compare the above algorithm with the one found in \cite{Ma00}.

One recursive step in the algorithm found in \cite{Ma00} consists  of 6 matrix multiplications and two  recursive calls.  
In this case only two matrix multiplications may be computed concurrently with the rest.  Therefore, the parallel 
implementation of one reccursive step of this algorithm consists of six sequential steps: four matrix multiplications 
and two recursive calls.

One recursive step in the algorithm described above consists  of 6 matrix multiplications and three recursive calls.  
In this case three matrix multiplications may be computed concurrently with the other multiplications and two, 
of the three, recursive calls may be executed concurrently.  Therefore, the parallel implementation of one reccursive 
step of the algorithm described above consists of five sequential steps: three matrix multiplications and two recursive 
calls.

This way, despite the fact that the algorithm described above has 50\% more operations, its depth is  25\% less. 

For the computations in both algorithms it is assumed  that the leading minors at every step are different from zero.  
If this assumption fails, we have to pivot rows, or columns, to make sure the leading minor is not zero.

The main difference of the new algorithm is that the choice of the nonzero leading minor is made independently within 
the local submatrix in each of the parallel branches.  After that, this submatrix, together with its (pivot) 
permutation matrix, is used in further computations.  We do not consider pivots in the above algorithm, as this 
will be the topic of another paper.

By contrast,  in the algorithm found in \cite{Ma00}, if a nonzero leading minor needs to be found  the whole 
computational process stops during its search.  Moreover, the search for this minor is done by the least diagonal 
block, and in case no pivot is found the search may be extended to the whole matrix. 

\bigskip
{ACKNOWLEDGMENTS}

\bigskip
Second author was supported in part by the Human Capital Foundation,
project no. 23-03-24, RFBR, project no. 04-07-90268
and program "Universities of Russia", project no. 04.01.464.


\end{document}